\documentclass[aps,pre,preprint,superscriptaddress,nofootinbib]{revtex4-1}
\usepackage{tensor}
\usepackage{amsmath}
\usepackage[dvipdfmx]{graphicx}
\usepackage{float}
\usepackage{caption}
\usepackage{comment}
\usepackage{amsfonts}
\usepackage{bm}
\usepackage{braket}
\usepackage{amsthm}
\usepackage{enumerate}

\newcommand{\sub}[1]{\ensuremath{_{\mathrm{#1}}}} 
\newcommand{\sps}[1]{\ensuremath{^{\mathrm{#1}}}}

\newtheorem{theorem}{Theorem}
\newtheorem{lemma}{Lemma}

\begin{document}
\title{Theorems on Entanglement Typicality in Non-equilibrium Dynamics}
\author{Koji Yamaguchi}
\affiliation{Graduate School of Science, Tohoku University,\\ Sendai, 980-8578, Japan}
%\date{\today}
 \begin{abstract}
  The notion of typicality in statistical mechanics is essential to
  characterize a macroscopic system. An overwhelming majority of the
  pure state looks almost identical if we neglect macroscopic non-local
  correlations, suggesting that thermal equilibrium is the collection of
  the typical properties. Quantum entanglement, which
  characterizes a non-local correlation, also has a typical behavior in
  equilibrium systems. However, it remains elusive whether there is a
  typical behavior of entanglement in dynamical non-equilibrium systems.
  To investigate the typicality, we consider a situation where a system
  in a pure state starts to share entanglement with its
  environment system due to the interaction between them. Assuming the
  initial state is randomly chosen from an ensemble of pure states, a criteria for
  the typicality of the R\'enyi entropies is presented. In addition, it
  is analytically proven that the second R\'enyi entropy has a typical
  behavior in two cases. The first one is an energy dissipation process
  in a multiple-qubit system which is initially in a random pure state
  in an energy shell. Since the typical behavior is qualitatively the
  same as the prediction of the Page curve conjecture, it gives the
  first proof of the Page curve conjecture in a dynamical process. In
  the second case, the typicality is proven for any dynamics described by a
  multiple-product of a single-qudit channel when the system is
  initially in a pure state randomly chosen from the whole Hilbert space. This result shows that entanglement typicality is not a specific feature of energy dissipating processes.
 \end{abstract}
\begin{flushright}
TU-1061
\end{flushright}
 \maketitle
\section{Introduction}
Typicality is one of the essential concepts in macroscopic
systems. The Sugita theorem \cite{Sugita} shows that macroscopic
observables composed of the sum of local operators cannot distinguish an
overwhelmingly majority of pure states in an energy shell, which implies
that if we neglect macroscopic non-local correlations, almost all of the
pure states look quite similar to each other. This mathematical fact
allows us to make a postulate that thermal equilibrium is the collection
of the typical behaviors \cite{Tasaki}. Then, what we practically
observe in an equilibrium system is predictable by using the
microcanonical ensemble average. Typicality gives an intuitive
explanation for the existence of the arrow of time in macroscopic
systems \cite{Lebowitz}. It should be noted that this interpretation of
thermal equilibrium is not directly related to the longstanding argument on the ergodicity in classical systems as a justification of the microcanonical ensemble. Actually,
it is known that the ``ergodic time'' is too long to give a sufficient
justification for a possible time-average in a realistic experiment \cite{Bricmont}. Since the Sugita theorem does not rely on the detailed dynamics of the system, it explains the reason why thermal equilibrium is so ubiquitous, and suggests that thermodynamic properties may be useful to investigate a system whose dynamics remains unknown. In fact, although black holes do not seem to be a normal thermodynamic system, it is known that they have thermodynamics-like laws \cite{Bardeen_Carter_Hawking}, and that they emit the Hawking radiation whose spectrum is thermal in the semi-classical approximation of the quantum gravity theory \cite{Hawking_rad}. 

There is another well-known typical feature in macroscopic equilibrium
systems --- typicality of entanglement. Entanglement characterizes a
non-local nature of the system, which is quantified by various kinds of
measures such as the entanglement entropy and the R\'enyi entropy. The
Lubkin--Lloyd--Pagels--Page (LLPP) theorem
\cite{Lubkin,Lloyd_Pagels,Page_sub} shows that a small subsystem
typically shares almost maximal entanglement with its complement system,
if the total macroscopic system is in a pure state randomly chosen from
the whole Hilbert space according to the Haar measure. This distribution
of state corresponds to the microcanonical ensemble for a system with a
trivial free Hamiltonian $H=0$, or a system at infinite temperature
where its free Hamiltonian is effectively negligible. Nevertheless, it
is not hard to expect that, in a system with a nontrivial Hamiltonian at
a finite temperature, the typical entanglement entropy is approximately
given by its maximum, i.e., the thermal entropy of the smaller
subsystem. A heuristic explanation is that the ordinary statistical
mechanical argument suggests the smaller subsystem is well described by
the Gibbs state, which is supported by the Sugita theorem. By using the
cTPQ state formulation \cite{Sugiura_Shimizu}, a universal feature of typical R\'enyi entropies at a finite temperature has been derived, and it is shown that the universal behavior is reproduced for eigenstates in a non-integrable system but not in an integrable system \cite {Fujita_Nakagawa_Sugiura_Watanabe}. In \cite{Vidmar_Hackl_Bianchi_Rigol}, it is shown that the average of entanglement entropy over the eigenstates in a translationally invariant system with quadratic fermionic Hamiltonian is different from the original result of the LLPP theorem. All these results  reveal the typical feature of
entanglement in an isolated equilibrium system in the sense that the total system is assumed to be in a pure state randomly chosen according to an ensemble.

Entanglement in non-equilibrium macroscopic systems is also of interest. In the context of quantum channel capacity, whose asymptotic theory essentially deals with a macroscopic system, it is known that entanglement can enhance the transmission rate for noisy channels \cite{Bennett_Shor_Smolin_Thapliyal}. 
Calculations on the time evolution of the entanglement entropy in
globally quenched systems show that there is a maximum speed of
propagation of the signal \cite{Calabrese_Cardy}, as expected from the causality. 
In \cite{Fujita_Nakagawa_Sugiura_Watanabe}, it has been checked that the
time evolution of the second R\'enyi entropy in a quenched system is well
fitted by the universal behavior derived in the paper for both a
non-integrable and an integrable system. In the context of the AdS/CFT
correspondence \cite{Maldacena}, it has been shown that the holographic
entanglement entropy for a very small subsystem has the first law-like
relationship for static and translationally invariant excited states
\cite{Bhattacharya_Nozaki_Takayanagi_Ugajin}. This result is extended to
the entanglement for spherical subsystems in a time-dependent excited
state \cite{Nozaki_Numasawa_Prudenziati_Takayanagi}. In
\cite{Fan_Zhang_Shen_Zhai}, it has been proven that the time evolution
of the second R\'enyi entropy in a quenched system is related with the
out-of-order correlation, which is important in the analysis of quantum
chaos. Besides these theoretical investigations, the second R\'enyi
entropy has been detected in recent ultracold atom experiments
\cite{EE_exp1,EE_exp2}, suggesting that the time evolution of macroscopic
entanglement may also be measured in future experiments.

In modern arguments on resolution scenarios for the black hole
information loss paradox, time evolution of entanglement between the
black hole and the Hawking radiation plays a significant role. In the
semi-classical approximation, it is shown that there remains the Hawking
radiation in a mixed state after the complete evaporation of a black
hole \cite{Hawking_BHIP}. Considering a black hole formed by a quantum
field in a pure state, it is often argued that this evaporation process
indicates the loss of quantum information and/or the violation of
unitarity. The mixedness of the Hawking radiation can be understood by
the entanglement between the Hawking radiation and the black hole. In
the semi-classical approximation, it is known that there is a partner
mode inside the black hole which purifies each Hawking mode
\cite{Hotta_Schutzhold_Unruh}, meaning that the entanglement
monotonically increases in time. This entanglement makes the Hawking
radiation mixed. If the quantum gravity theory is unitary, there must be the purification partner of the Hawking radiation. Possible known candidates for the partner are the Hawking radiation itself \cite{Page_curve,Almheiri_Marolf_Polchinski_Sully,Almheiri_Marolf_Polchinski_Stanford_Sully}, remnants \cite{Aharonov_Casher_Nussinov}, zero-point fluctuation \cite{Wilczek,Hotta_Schutzhold_Unruh}, degrees of freedom in baby universe \cite{Unruh_Wald} and soft hairs \cite{Hawking_Perry_Strominger}.  

In the pioneering work done by Page \cite{Page_curve}, it is conjectured
that a time evolution of the entanglement entropy between a black hole
and the Hawking radiation is given by the typical behavior derived by
using the LLPP theorem. Under the assumption that the total system is
always in a random pure state, and that the evaporation process is
described by the change in the dimensions of the Hilbert spaces for the
Hawking radiation and the black hole, the typical entanglement entropy
is given by the thermal entropy of the smaller subsystem. This evolution
curve of entanglement is called the Page curve. One interesting
aspect of its prediction is that it predicts that the entanglement
entropy after the half-evaporation time, called the Page time,
is given by the Bekenstein--Hawking entropy \cite{Bekenstein,
Hawking_rad}. It is difficult to check the validity of the conjecture in
a realistic black hole evaporation process since the quantum gravity
theory has not been established completely. For black holes with
negative heat capacity,  the composite system of the black hole and the
Hawking radiation has an instability \cite{Hawking_instability}, meaning
that the Page curve conjecture seems to be doubtful since the total
system cannot be in a thermal equilibrium. In an analysis on a black
hole evaporation qubit model with negative heat capacity
\cite{Hotta_Nambu_Yamaguchi}, it is pointed out that the emission of
soft hair may make the entanglement
larger than the Page curve, which would avoid the emergence of the firewall
\cite{Almheiri_Marolf_Polchinski_Sully,Almheiri_Marolf_Polchinski_Stanford_Sully}. It
should be noted that the Page curve conjecture does not rely on any
detailed dynamics of black holes. Thus, an analysis in a condensed
matter system has an implication for the validity of the conjecture.

The main objective of this paper is not to discuss the validity of the
Page curve in black hole evaporation processes but to investigate a
typical entanglement in a dynamical process from a general
perspective. In a macroscopic system, it is difficult to calculate the
time evolution of the entanglement exactly, even when one knows the
initial state. Rather, the success of the equilibrium statistical
mechanics suggests that typical features have fundamental importance. We
prepare the initial state as a pure state randomly chosen from a subset
of the total Hilbert space. The system interacts with the environment
system and starts to share entanglement. We derive a general formula
for the average and the standard deviation of the trace of the $m$-th
power of the state to investigate the typicality of $m$-th R\'enyi
entropy. Once one could show that the ratio of the standard deviation to
the average is exponentially small as the system size becomes large, Chebyshev's inequality ensures the probability of getting an atypical
value is also exponentially small, meaning that we practically observe
the average value in a macroscopic system. Although it is quite
difficult to prove the typicality for general dynamics, we provide two
cases where the typicality of the second R\'enyi entropy can be proven
analytically. The first one is an energy dissipation process in a
multiple-qubit system, which is similar to the situation in the Page
curve conjecture. We show that at the very last stage of the dissipation
process, the typical value behaves similar to the R\'enyi entropy
calculated by the Gibbs state. This is the ``first proof'' of the Page curve
conjecture as a dynamical process. In the second case, we investigate
the typicality for a multiple-qudit system under the assumption that the
system is initially in a pure state distributed over the whole Hilbert
space. It is shown that the second R\'enyi entropy has a typical
behavior for any process composed of a tensor product of a single-qudit
quantum channel, meaning that the typicality of entanglement is not a
specific feature of an energy dissipation process.

This paper is organized as follows. In Section
\ref{sec_average_variance}, we derive a formula for the average and
the standard deviation for the trace of $m$-th power of state. In addition, a general
magnitude relationship between them is proven. In Section
\ref{sec_proof_typicality}, we analytically prove the typicality of the
second R\'enyi entropy for two kinds of cases. In Section
\ref{sec_conclusion}, we summarize our conclusions. 

\section{Ensemble average and Variance}\label{sec_average_variance}
Let us consider a system $A$ in a pure state
$\ket{\Psi}_A\in\mathcal{H}_A$ at the initial time. When the system $A$
interacts with the environment system $E$, these two subsystems start to
share entanglement. The most general time evolution of the system $A$
can be described by a quantum channel $\mathcal{E}_A$, and it will
evolve into
$\rho_A\equiv\mathcal{E}_A\left(\ket{\Psi}_A\bra{\Psi}_A\right)$. In
general, the output state $\rho_A$ can be a state for a system different
from $A$. Just for simplicity, let us assume the output system is the
same as the original one. It is easy to extend the results below for a nontrivial output system. Entanglement between the system $A$ and its environment can be quantified by using the $m$-th R\'enyi entropy $R_m$ defined by
\begin{align}
 R_m\equiv \frac{1}{1-m}\ln{\left(\mathrm{Tr}\left(\rho_A^m\right)\right)}.
\end{align}
R\'enyi entropy is related with the entanglement entropy $S\sub{EE}\equiv- \mathrm{Tr}\left(\rho_A\ln{\rho_A}\right)=\lim_{m\to 1}R_m$.

In order to investigate the typicality of entanglement, let us consider
an ensemble of the initial state $\ket{\Psi}_A$ which is distributed
according to the Haar measure of the unitary group on a sub-Hilbert
space $\mathcal{H}_S\subset \mathcal{H}_A$. The most important example
of ensemble is the microcanonical ensemble, where the state is chosen
uniformly from an energy shell. Introducing an orthonormal basis of the
$\mathcal{H}_S$ as $\{\ket{i}_A\}_{i=1}^{d_S}$, a Haar-random state can
be expanded as $\ket{\Psi}_A=\sum_{i=1}^{d_S}X_i
\ket{i}_A\in\mathcal{H}_S$, where $d_S\equiv\dim{\mathcal{H}_S}$ and
$\{X_i\}$ are random coefficients. The detailed definition of $\{X_i\}$
is given in Appendix \ref{ensemble}. Then, 
\begin{align}
 \rho_A(\{X_i\})=\sum_{i,j=1}^{d_S}X_iX_j^*\mathcal{E}_A\left(\ket{i}_A\bra{j}_A\right)
\end{align}
describes the state evolved by $\mathcal{E}_A$.
The ensemble average of $R_m$ is given by
\begin{align}
 \overline{R_m}&=\frac{1}{1-m}\overline{\ln{\left(\mathrm{Tr}\left(\rho_A\left(\{X_i\}\right)^m\right)\right)}}\nonumber\\
 &=\frac{1}{1-m}\overline{\ln{\left(\mathrm{Tr}\left(\prod_{k=1}^m \sum_{i^{(k)},j^{(k)}=1}^{d_S}X_{i^{(k)}}X_{j^{(k)}}^*\mathcal{E}_A\left(\ket{i^{(k)}}\bra{j^{(k)}}\right)\right)\right)}},
\end{align}
where the overline denotes the integration over the Haar measure.
It is difficult to calculate the average of logarithm of random coefficients ${X_i}$ directly. Instead, let us use the following quantity:
\begin{align}
 \tilde{R}_m\equiv\frac{1}{1-m}\ln{\left(\overline{\mathrm{Tr}\left(\rho_A\left(\{X_i\}\right)^m\right)}\right)},
\end{align}
which is obtained by calculating the average of the trace of $m$-th power of $\rho\left(\left\{X_i\right\}\right)$:
\begin{align}
 \overline{\mathrm{Tr}\left(\rho_A\left(\{X_i\}\right)^m\right)}=\overline{\prod_{k=1}^m \sum_{i^{(k)},j^{(k)}=1}^{d_S}X_{i^{(k)}}X_{j^{(k)}}^*}\mathrm{Tr}\left(\prod_{l=1}^m\mathcal{E}_A\left(\ket{i^{(l)}}\bra{j^{(l)}}\right)\right).
\end{align}
If the standard deviation divided by the average of
$\mathrm{Tr}\left(\rho_A\left(\{X_i\}\right)^m\right)$ is exponentially
small in $N$, then the typical value $\overline{R_m}$ is approximately
given by $\tilde{R}_m$. A similar argument can also be found in \cite{Fujita_Nakagawa_Sugiura_Watanabe}.

As is derived in Appendix \ref{ensemble}, 
\begin{align}
 \overline{\prod_{k=1}^m X_{i^{(k)}}X_{j^{(k)}}^*}=\frac{\Gamma(d_S)}{\Gamma(d_S+m)}\sum_{\tau\in S_m}\prod_{k=1}^m \delta_{j^{(k)},i^{(\tau(k))}},
\end{align}
where $S_m$ denotes the symmetric group of degree $m$. By using this
formula, the average and the variance of
$\mathrm{Tr}\left(\rho_A\left(\{X_i\}\right)^m\right)$ for arbitrary
positive integer $m$ can be obtained, although they are complicated. The average is given by
\begin{align}
 \overline{\mathrm{Tr}\left(\rho_A\left(\{X_i\}\right)^m\right)}=\frac{\Gamma(d_S)}{\Gamma(d_S+m)}\sum_{\tau\in S_m}\sum_{i^{(1)},\cdots, i^{(m)}=1}^{d_S}\mathrm{Tr}\left(\prod_{k=1}^m\mathcal{E}_A\left(\ket{i^{(k)}}\bra{i^{(\tau(k))}}\right)\right).\label{eq_ave}
\end{align}
On the other hand, 
\begin{align}
 & \overline{\left(\mathrm{Tr}\left(\rho_A\left(\{X_i\}\right)^m\right)\right)^2}\nonumber\\
&=\frac{\Gamma(d_S)}{\Gamma(d_S+2m)}\sum_{\tau\in S_{2m}}\sum_{i^{(1)},\cdots,i^{(2m)}=1}^{d_S}\mathrm{Tr}\left(\prod_{k=1}^m\mathcal{E}_A\left(\ket{i^{(k)}}\bra{i^{(\tau(k))}}\right)\right)\mathrm{Tr}\left(\prod_{l=m+1}^{2m}\mathcal{E}_A\left(\ket{i^{(l)}}\bra{i^{(\tau(l))}}\right)\right).\label{eq_sq_ave}
\end{align}
$S_{2m}$ contains a subset $S_m\times S_m$ which is composed of a product of elements in $S_m$ in the sense that
\begin{align}
 \begin{pmatrix}
  1&2&\cdots&m &m+1&m+2&\cdots&2m\\
  \tau_1(1)&\tau_1(2)&\cdots&\tau_1(m) &m+\tau_2(1)&m+\tau_2(2)&\cdots&m+\tau_2(m)
 \end{pmatrix}
\in S_{2m},
\end{align}
where $\tau_1,\tau_2\in S_m$. Thus, the ratio of Eq~(\ref{eq_sq_ave}) to the square of Eq~(\ref{eq_ave}) is given by a product of
\begin{align}
 \frac{\frac{\Gamma(d_S)}{\Gamma(d_S+2m)}}{\left(\frac{\Gamma(d_S)}{\Gamma(d_S+m-1)}\right)^2}=\frac{d_S(d_S+1)\cdots(d_S+m-1)}{(d_S+m)(d_S+m+1)\cdots (d_S+2m-1)}
\end{align}
and 
\begin{align}
 1+\frac{\sum_{\tau\in S_{2m}\backslash S_m\times S_m}\sum_{i^{(1)},\cdots,i^{(2m)}=1}^{d_S}\mathrm{Tr}\left(\prod_{k=1}^m\mathcal{E}_A\left(\ket{i^{(k)}}\bra{i^{(\tau(k))}}\right)\right)\mathrm{Tr}\left(\prod_{l=m+1}^{2m}\mathcal{E}_A\left(\ket{i^{(l)}}\bra{i^{(\tau(l))}}\right)\right)}{\left(\sum_{\tau\in S_m}\sum_{i^{(1)},\cdots, i^{(m)}=1}^{d_S}\mathrm{Tr}\left(\prod_{k=1}^m\mathcal{E}_A\left(\ket{i^{(k)}}\bra{i^{(\tau(k))}}\right)\right)\right)^2}.\label{eq_second_factor}
\end{align}
If $d_S$ grows exponentially fast as the system size becomes large, the
first factor is given by $1$ plus an exponentially small term. For
example, if we take $\mathcal{H}_S$ as an Hilbert space spanned by
energy eigenstates in an energy shell, this condition is satisfied for a
normal thermodynamic system. 
Since the variance is given by
\begin{align}
 \overline{\left(\mathrm{Tr}\left(\rho_A\left(\{X_i\}\right)\right)-\overline{\mathrm{Tr}\left(\rho_A\left(\{X_i\}\right)\right)}\right)^2}= \overline{\left(\mathrm{Tr}\left(\rho_A\left(\{X_i\}\right)^2\right)\right)^2}-\left(  \overline{\mathrm{Tr}\left(\rho_A\left(\{X_i\}\right)^2\right)}\right)^2,
\end{align}
if the second term in Eq~(\ref{eq_second_factor}) is exponentially small, the ratio of the standard deviation to the average is also exponentially small, meaning that the R\'enyi entropy $R_m$ has a typical behavior.

Hereafter, let us only investigate the lowest order R\'enyi entropy $R_2$. In fact, for $m=2$, there is a useful expression for the average and the variance. From Eq.~(\ref{eq_ave}), the average value is given by
\begin{align}
 \overline{\mathrm{Tr}\left(\rho_A\left(\{X_i\}\right)^2\right)}=\frac{1}{d_S(d_S+1)}\sum_{\tau\in S_2}\sum_{i^{(1)},i^{(2)}=1}^{d_S}\mathrm{Tr}\left(\mathcal{E}_A\left(\ket{i^{(1)}}\bra{i^{(\tau(1))}}\right)\mathcal{E}_A\left(\ket{i^{(2)}}\bra{i^{(\tau(2))}}\right)\right).
\end{align}
Define an orthonormal basis for the set of traceless Hermite operators on $\mathcal{H}_A$ as $\{\sigma_\mu\}_{\mu=1}^{d_A^2-1}$ satisfying $\mathrm{Tr}\left(\sigma_\mu^\dag\sigma_\nu\right)=\delta_{\mu,\nu}$. Any linear operator $L$ on $\mathcal{H}_A$ can be expanded as
\begin{align}
 L=\sum_{\mu=0}^{d_A^2-1}\mathrm{Tr}\left(\sigma_\mu L\right)\sigma_\mu,
\end{align}
where we have defined $\sigma_0\equiv\frac{1}{\sqrt{d_A}}\mathbb{I}_A$ and $\mathbb{I}_A$ is the identity operator on the Hilbert space $\mathcal{H}_A$. By using a set of Kraus operators $\{K_i\}$ for the quantum channel $\mathcal{E}_A$, let us introduce a set of Hermite operators $\{Q_\mu\}_{\mu=0}^{d_A^2-1}$ as
\begin{align}
 Q_\mu \equiv \sum_{i}K_i^\dag \sigma_\mu K_i.
\end{align}
Since
\begin{align}
 \mathcal{E}_A\left(\ket{i}\bra{j}\right)=\sum_{\mu=0}^{d_A^2-1}\mathrm{Tr}\left(\mathcal{E}\left(\ket{i}\bra{j}\right)\sigma_\mu\right)\sigma_\mu=\sum_{\mu=0}^{d_A^2-1}\braket{j|Q_\mu|i}\sigma_\mu,
\end{align}
we get
\begin{align}
 \mathrm{Tr}\left(\mathcal{E}_A\left(\ket{i^{(1)}}\bra{i^{(\tau(1))}}\right)\mathcal{E}_A\left(\ket{i^{(2)}}\bra{i^{(\tau(2))}}\right)\right)=\sum_{\mu=0}^{d_A^2-1}\braket{i^{(\tau(1))}|Q_\mu|i^{(1)}}\braket{i^{(\tau(2))}|Q_\mu|i^{(2)}}.
\end{align}
Therefore, 
\begin{align}
  \overline{\mathrm{Tr}\left(\rho_A\left(\{X_i\}\right)^2\right)}&=\frac{1}{d_S(d_S+1)}\left(\alpha_1+\alpha_2\right) 
\end{align}
where we have defined $\alpha_1\equiv\sum_{\mu=0}^{d_A^2-1}\mathrm{Tr}\left(R_\mu\right)^2$, $\alpha_2\equiv\sum_{\mu=0}^{d_A^2-1}\mathrm{Tr}\left(R_\mu^2\right)$, $R_\mu\equiv P Q_\mu P$ and $P\equiv \sum_{i=1}^{d_S}\ket{i}_A\bra{i}_A$ is the projection operator onto the sub-Hilbert space $\mathcal{H}_S$. In a very similar way, a straightforward calculation shows that
\begin{align}
 \overline{\left(\mathrm{Tr}\left(\rho_A\left(\{X_i\}\right)^2\right)\right)^2}=\frac{1}{d_S(d_S+1)(d_S+2)(d_S+3)}\left(\left(\alpha_1+\alpha_2\right)^2+2\beta_1+4\beta_2+8\beta_3+4\beta_4+2\beta_5\right),
\end{align}
where
\begin{equation}
\begin{split}
  \beta_1&\equiv \sum_{\mu=0}^{d_A^2-1}\sum_{\nu=0}^{d_A^2-1}\left(\mathrm{Tr}\left(R_\mu R_\nu\right)\right)^2,\\
\beta_2&\equiv\sum_{\mu=0}^{d_A^2-1}\sum_{\nu=0}^{d_A^2-1}\mathrm{Tr}\left(R_\mu R_\nu\right)\mathrm{Tr}\left(R_\mu\right)\mathrm{Tr}\left(R_\nu\right),\\
 \beta_3&\equiv \sum_{\mu=0}^{d_A^2-1}\sum_{\nu=0}^{d_A^2-1} \mathrm{Tr}\left(R_\mu R_\nu^2\right)\mathrm{Tr}\left(R_\mu\right),\\
 \beta_4&\equiv \sum_{\mu=0}^{d_A^2-1}\sum_{\nu=0}^{d_A^2-1}\mathrm{Tr}\left(R_\mu^2 R_\nu^2\right), \\
 \beta_5&\equiv  \sum_{\mu=0}^{d_A^2-1}\sum_{\nu=0}^{d_A^2-1}\mathrm{Tr}\left(R_\mu R_\nu R_\mu R_\nu\right). 
\end{split}
\end{equation}
To consider the thermodynamic limit, let us assume $\mathcal{H}_A$ is
composed of $N(\gg 1)$ copies of a small subsystem. For a normal thermodynamic system $d_S$ grows exponentially fast in $N$.  
If the ratio of $\max_{u=1,\cdots, 5}\{|\beta_u|\}$ to
$\max_{i,j=1,2}\{|\alpha_i\alpha_j|\}$ is exponentially small in $N$,
then one can conclude that the R\'enyi entropy has a typical behavior which is approximately given by
\begin{align}
 \overline{R_2}\approx -\ln{\left(\frac{1}{d_S(d_S+1)}(\alpha_1+\alpha_2)\right)}
\end{align}
as the system size become large, i.e., $N\to\infty$. 

For the magnitude relationships between $\max_{u=1,\cdots, 5}\{|\beta_u|\}$ and $\max_{i,j=1,2}\{|\alpha_i\alpha_j|\}$, the following theorem holds:
\begin{theorem}\label{thm}
 For any quantum channel $\mathcal{E}_A$ and the sub-Hilbert space
 $\mathcal{H}_S\subset \mathcal{H}_A$,
\begin{equation}
\max_{u=1,\cdots, 5}\{|\beta_u|\}<\max_{i,j=1,2}\{|\alpha_i\alpha_j|\}
\end{equation}
holds if $\dim{\mathcal{H}_S}>1$.
\end{theorem}
The proof is given in Appendix \ref{thm_proof}. It should be noted that if $\dim{\mathcal{H}_S}=1$, then the variance is exactly zero.
Thus, the exception in the theorem is not significant here. Although
this theorem cannot be used to directly evaluate the asymptotic behavior
of the ratio of $\max_{u=1,\cdots, 5}\{|\beta_u|\}$ to
$\max_{i,j=1,2}\{|\alpha_i\alpha_j|\}$, it suggests that there is a
typical second R\'enyi entropy when they have different exponential behaviors.

The results presented in this section hold for any quantum channel
$\mathcal{E}_A$ and ensemble defined by using the Haar measure for the
unitary group on a set
of initial pure states. For the microcanonical distribution for initial
states in a normal thermodynamical system, $d_S$ grows exponentially
in $N$. On the other hand, it is generally difficult to
analytically evaluate the asymptotic behaviors of $\alpha$s and $\beta$s
since $\{R_\mu\}$ will be complicated. In the rest of this paper, we restrict ourself to the
cases where the dynamics is described by a multiple-product of an
independent and identical channel $\mathcal{E}$, i.e., $\mathcal{E}_A=\mathcal{E}^{\otimes N}$. This assumption significantly simplifies the calculation of $\alpha$s and $\beta$s. In the following section, we analytically prove the typicality in two cases for this class of channels.

\section{Typicality of the second R\'enyi entropy} \label{sec_proof_typicality}
\subsection{Energy dissipation process and the Page curve conjecture}\label{dissipation}
The Page curve conjecture originally stems from the black hole
evaporation process where a black hole loses its energy by emitting the
Hawking radiation. After the Page time, the typical entanglement entropy
is given by the thermal entropy of the emitter, i.e., the black hole.
This work is fascinating because the result is universal in the sense
that the typical entanglement does not seem to depend on the details of
the dynamics nor the Hilbert space of the black hole. Therefore, it is
sometimes believed that the evolution of the entanglement entropy in
black hole evaporation processes in the quantum gravity theory follows
this conjecture. However, the dynamics assumed in the conjecture seems
to be too simple and ambiguous since we usually do not regard the change in the dimensions of the Hilbert spaces as an energy dissipation process.

In this subsection, we prove the typicality of the second R\'enyi entropy in multiple-qubit system for an energy dissipation process by calculating $\alpha$s and $\beta$s for the amplitude damping channel. We do not claim that it describes the black hole evaporation process. Nevertheless, it would be interesting to investigate the relationships between the asymptotic behavior of the typical R\'enyi entropy and the Page curve conjecture, since the conjecture does not rely on detailed dynamics of the black hole. The result shows that the typical R\'enyi entropy appears to be in line with the Page curve conjecture in the very last stage of the evaporation, in the sense that the asymptotic behavior gives the same value as the second R\'enyi entropy calculated by the Gibbs state.

Consider an $N$ copies of a qubit system whose free Hamiltonian is given by
\begin{align}
 H=\omega \ket{+}\bra{+},
\end{align}
where $\omega>0$ is a positive constant. In this set up, the
microcanonical ensemble specified by an energy shell $[E_0,E_0+\Delta)$ is characterized by the Hilbert space $\mathcal{H}_S$, which is spanned by tensor product vectors of $K\equiv E_0/\omega$ $\ket{+}$ states and $N-K$ $\ket{-}$ states, where $\{\ket{+},\ket{-}\}$ is a set of orthonormal basis for a qubit. Here we have assumed $K$ is a positive integer and $\Delta/\omega\in(0,1)$ for simplicity. The thermodynamic limit is given by $N\to \infty$ with $u\equiv \frac{K}{N}$ fixed. Then, $d_S=\binom{N}{K}$ is exponentially large in $N$ and
\begin{align}
 P=K!\frac{d^K}{dx^K}\bigotimes_{n=1}^N\left(x \ket{+}_n\bra{+}_n+\ket{-}_n\bra{-}_n\right)\Bigg|_{x=0}.
\end{align} 

As a simplest energy dissipation process, we assume each single qubit evolves through the amplitude damping channel $\mathcal{E}$ characterized by the Kraus operators
\begin{align}
 K_1=\ket{-}\bra{-}+\sqrt{1-r}\ket{+}\bra{+},\quad K_2= \sqrt{r}\ket{-}\bra{+},
\end{align}
where $r\in[0,1]$ describes the probability of emission.
Then, the dynamics of the total qubit system is described by $\mathcal{E}_A=\mathcal{E}^{\otimes N}$. Rescaled Pauli matrices for a single qubit defined by
\begin{equation}
\begin{split}
 \sigma_0^{(1\text{-qubit})}&\equiv\frac{1}{\sqrt{2}}\left(\ket{+}\bra{+}+\ket{-}\bra{-}\right)\\
 \sigma_1^{(1\text{-qubit})}& \equiv \frac{1}{\sqrt{2}}\left(\ket{+}\bra{-}+\ket{-}\bra{+}\right)\\
 \sigma_2^{(1\text{-qubit})}&\equiv \frac{i}{\sqrt{2}}\left(-\ket{+}\bra{-}+\ket{-}\bra{+}\right) \\
 \sigma_3^{(1\text{-qubit})}&\equiv \frac{1}{\sqrt{2}}\left(\ket{+}\bra{+}-\ket{-}\bra{-}\right) 
\end{split}
\end{equation}
can be used to construct an orthonormal basis of the set of linear operators as
\begin{align}
 \sigma_{\mu}=\bigotimes_{n=1}^N\sigma_{\mu_n}^{(1\text{-qubit})},
\end{align}
where we have identified $\mu=0,1,\cdots 2^N-1$ and its $N$-length binary representation
$\mu=\mu_1\mu_2\cdots\mu_N$. As is easily checked, these operators satisfy the following normalization condition:
\begin{align}
 \mathrm{Tr}\left(\sigma_\mu\sigma_\nu\right)=\delta_{\mu,\nu}\equiv \prod_{n=1}^N\delta_{\mu_n,\nu_n}.
\end{align}
Then, the set of operator $Q_\mu$ is factorized as
\begin{align}
 Q_\mu=\bigotimes_{n=1}^NQ_{\mu_n}^{(1\text{-qubit})},
\end{align}
where
\begin{align}
 Q_{\mu_n}^{(1\text{-qubit})}\equiv \sum_{i=1}^2 K_i^\dag \sigma_{\mu_n}^{(1\text{-qubit})} K_i.
\end{align}
Therefore,
\begin{align}
 PQ_\mu=K!\frac{d^K}{dx^K}\bigotimes_{n=1}^N \tilde{Q}_{\mu_n}^{(1\text{-qubit})}(x)\big|_{x=0},
\end{align}
where
\begin{align}
 \tilde{Q}_{\mu_n}^{(1\text{-qubit})}(x)\equiv(x\ket{+}_n\bra{+}_n+\ket{-}_n\bra{-}_n)Q_{\mu_n}^{(1\text{-qubit})}
\end{align}

By using this expression, $\alpha_1,\alpha_2$ can be evaluated easily. For example,
\begin{align}
 \alpha_1&=\sum_{\mu}\mathrm{Tr}\left(R_\mu\right)^2\nonumber\\
 &=(K!)^2\frac{d^K}{dx_1^K}\frac{d^K}{dx_2^K}\left(\sum_{\mu=0}^3\mathrm{Tr}\left(\tilde{Q}_{\mu}^{(1\text{-qubit})}(x_1)\right) \mathrm{Tr}\left(\tilde{Q}_{\mu}^{(1\text{-qubit})}(x_2)\right)\right)^N \Bigg|_{x_1,x_2=0}.
\end{align}
Straightforward calculation shows
\begin{align}
 \sum_{\mu=0}^3\mathrm{Tr}\left(\tilde{Q}_{\mu}^{(1\text{-qubit})}(x_1)\right) \mathrm{Tr}\left(\tilde{Q}_{\mu}^{(1\text{-qubit})}(x_2)\right)=(r^2+(1-r)^2)x_1x_2 +r(x_1+x_2)+1,
\end{align}
and by reading off the coefficient of $x_1^Kx_2^K$ in $\left(
\sum_{\mu=0}^3\mathrm{Tr}\left(\tilde{Q}_{\mu}^{(1\text{-qubit})}(x_1)\right)
\mathrm{Tr}\left(\tilde{Q}_{\mu}^{(1\text{-qubit})}(x_2)\right)\right)^N$,
we get
\begin{align}
 \alpha_1&=\binom{N}{K}\sum_{k=0}^K\binom{K}{k}\binom{N-K}{K-k}(r^2+(1-r)^2)^k r^{2(K-k)}\label{alpha_1_def}\\
 &=\frac{N!} {K!^2(N-2K)!}r^{2K}{}_2 F_1\left(-K,-K;N-2K+1;\frac{r^2+(1-r)^2}{r^2}\right),
\end{align}
where ${}_2 F_1\left(a,b;c;z\right)$ is the hypergeometric function defined by $ {}_2 F_1\left(a,b;c;z\right)\equiv \sum_{k=0}^\infty\frac{(a)_k(b)_k}{(c_k)}\frac{z^k}{k!}$ with the rising Pochhammer symbol $(a)_k\equiv \Gamma(a+k)/\Gamma(a)$. Similarly, it is shown
\begin{align}
 \alpha_2=\frac{N!} {K!^2(N-2K)!} (1-r)^{2K} {}_2F_1\left(-K,-K;N-2K+1;\frac{r^2+(1-r)^2}{(1-r)^2}\right).
\end{align}
If we interchange $r$ into $1-r$, $\alpha_1$ turns into $\alpha_2$ and vice versa. As is seen from Eq.~(\ref{alpha_1_def}), $\alpha_1<\alpha_2$ for $r<1/2$, $\alpha_1=\alpha_2$ for $r=1/2$, and $\alpha_1>\alpha_2$ for $r>1/2$.  As is shown in Appendix \ref{asy_appendix}, the asymptotic behavior of the hypergeometric function ${}_2F_1\left(-Nu,-Nu;(1-2u)N+1;z\right)$ for large $N$ is given by $\exp{\left(Ng_u(z)\right)}$, where
\begin{align}
  g_u(z)&\equiv
 -\left(\frac{1}{2}-u\right)\ln{\left(\frac{2+\frac{4u(1-u)}{(1-2u)^2}z+2\sqrt{1+\frac{4u(1-u)}{(1-2u)^2}z}}{4}\right)}\nonumber\\
 &\qquad\qquad +\frac{1}{2}\ln{\left(\frac{2+\frac{4u(1-u)}{(1-2u)^2}(1+z)+\frac{2}{1-2u}\sqrt{1+\frac{4u(1-u)}{(1-2u)^2}z}}{4\left(\frac{1-u}{1-2u}\right)^2}\right)}.
\end{align}
Therefore, for large $N$, $\alpha$s behave like $\alpha_i\sim e^{Nh_i(r)}$, where
\begin{equation}
 h_i(r)\equiv
\begin{cases}
 -2u\ln{u}-(1-2u)\ln{(1-2u)}+2u\ln{r}+g_u\left(\frac{r^2+(1-r)^2}{r^2}\right)&\quad (i=1)\\
-2u\ln{u}-(1-2u)\ln{(1-2u)}+ 2u\ln{(1-r)}+g_u\left(\frac{r^2+(1-r)^2}{(1-r)^2}\right) &\quad (i=2)
\end{cases}.
\end{equation}
Fig.~\ref{asy_fig} shows the plots of $h_1(r)$ and $\frac{1}{N}\ln{\alpha_1}$ with different $N$s, which verifies the asymptotic behavior.
\begin{figure}[tbp]
 \includegraphics[width=7.5cm]{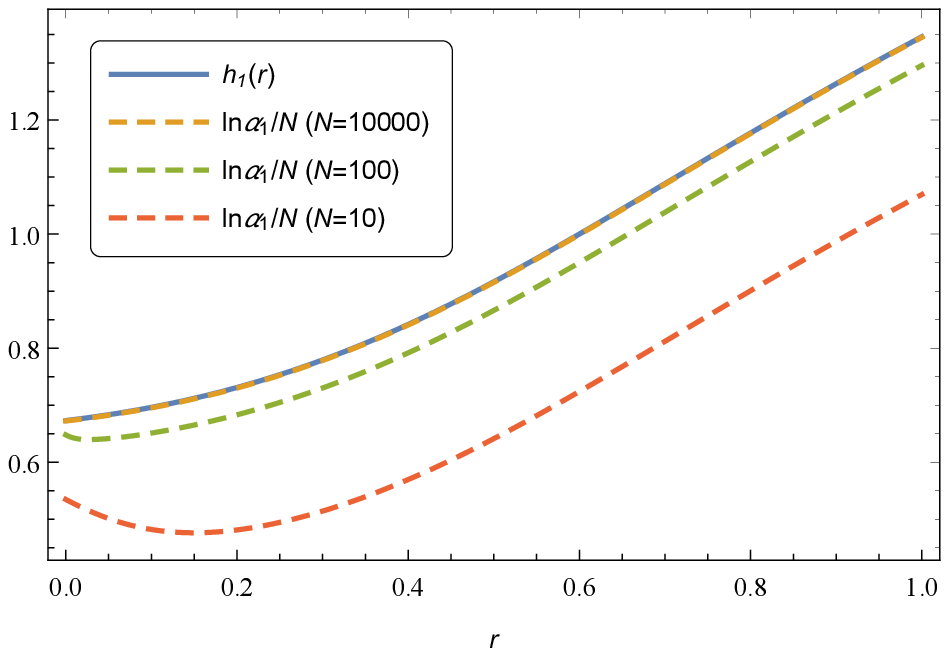}
\caption{The plots of functions for $u=0.4$, showing the validity of the asymptotic behavior. It is seen that as $N$ becomes large, $\ln{\alpha_1}/N\sim h_1$, meaning that $\alpha_1\sim e^{Nh_1}$. Since $\alpha_2(r)=\alpha_1(1-r)$, the asymptotic behavior of $\alpha_2$ is given by $e^{Nh_1(1-r)}=e^{Nh_2(r)}$.}\label{asy_fig}
\end{figure}
Fig.~\ref{hs} shows the behaviors of $h_1(r)$ and $h_2(u)$. From this figure, one can see that $\alpha_1\sim\alpha_2$ only when $r=1/2$ for $u\in(0,1/2)$.
\begin{figure}[tbp]
\includegraphics[width=7.5cm]{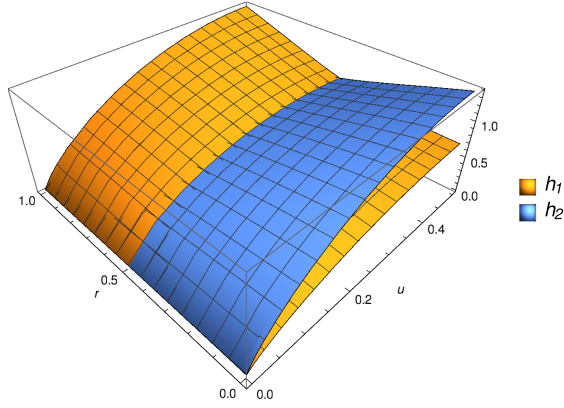}
\caption{The comparison of $h_1$ and $h_2$. This plot shows that
 $h_1=h_2$ only if $r=1/2$ for $u\in(0,1/2)$.}
\label{hs}
\end{figure}

$\beta$s are also obtained in the same way, and the results are
\begin{equation}
 \begin{split}
 \beta_1&=\sum_{\substack{k,l_1,l_2\geq
 0\\k+l_1+l_2=K}}\frac{N!}{l_1!^2l_2!^2k!(N-2K+k)!}(r^2+(1-r)^2)^{2k}r^{4l_1}(1-r)^{4l_2},\\
 \beta_2&=\frac{1}{d_S}\alpha_1^2, \\
 \beta_3&=\frac{1}{d_S}\alpha_1\alpha_2 ,\\
 \beta_4&=\frac{1}{d_S}\alpha_2^2 ,\\
 \beta_5&=\sum_{\substack{k,l_1,l_2\geq
 0\\k+l_1+l_2=K}}\frac{N!}{l_1!^2l_2!^2k!(N-2K+k)!}\left(r^2+(1-r)^2\right)^{2k}r^{2(l_1+l_2)}(1-r)^{2(l_1+l_2)}.  
 \end{split}
\end{equation}
In this form, it might seem to be difficult to compare $\beta_1$ and
$\beta_5$ with $\alpha_1^2$ and $\alpha_2^2$. However, as is derived in
Appendix \ref{thm_proof}, $|\beta_1|\leq| \alpha_2|^2$ and
$|\beta_5|\leq |\beta_4|$ hold for any channel. Furthermore, if we interchange $r$ into $1-r$, $\alpha_2$ turns into $\alpha_1$ while $\beta_1$ is invariant. Thus, $|\beta_1|\leq \min_{i,j=1,2}{\{|\alpha_i\alpha_j|\}}$.

Since the asymptotic behaviors of $\alpha$s coincide with each other
only when $r=1/2$, $|\beta_1|/\max_{i,j=1,2}{\{|\alpha_i\alpha_j|\}}$ is
exponentially small except for $r=1/2$. For $r=1/2$,
$\beta_1(1/2)=\beta_5(1/2)\leq \beta_4(1/2)=\alpha_2^2/d_S$.
Therefore, the ratio of $\max_{u=1,\cdots,5}\{{|\beta_u|}\}$ to
$\max_{i,j=1,2}{\{|\alpha_i\alpha_j|\}}$ is exponentially small for any
$r\in[0,1]$, meaning that we have proven the following theorem:
\begin{theorem}
 Consider an $N$-qubit system $A$, where the energy difference between
 the excited state and the ground state is the same for each qubit. Under the assumption that the system
 is initially in a random pure state in an energy shell and that each qubit
 evolves according to the amplitude damping channel, the second R\'enyi
 entropy between the system $A$ and its environment is typical when $N\to \infty$.
\end{theorem}

It would be interesting to compare the typical value with the Page-like curve for the second R\'enyi entropy. The total energy for the $N$-qubit system is $E(r)=Nu\omega(1-r)$. The corresponding Gibbs state is given by
\begin{align}
 \rho\sub{Gibbs}(r)=\bigotimes_{n=1}^N\left(u(1-r)\ket{+}_n\bra{+}_n+\left(1-u(1-r)\right)\ket{-}_n\bra{-}_n\right).
\end{align}
The second R\'enyi entropy per qubit defined by
\begin{align}
 \frac{1}{N}R_{2}\sps{thermal}(r)\equiv-\frac{1}{N}\ln{\left(\mathrm{Tr}\left(\rho\sub{Gibbs}{}_{,u}(r)^2\right)\right)}=-\ln{ \left(u^2(1-r)^2+\left(1-u(1-r)\right)^2\right)}.
\end{align}
satisfies 
\begin{align}
  \frac{1}{N}R_{2}\sps{thermal}(r=1)=0,\quad \frac{d}{d r}\left(\frac{1}{N}R_{2}\sps{thermal}(r)\right)\Biggr|_{r=1}=-2u,
\end{align}
where $r=1$ is the complete evaporation point.
On the other hand,
\begin{align}
 \frac{1}{N}\tilde{R}_2(r=1)= 0,\quad\frac{d}{dr}\left(\frac{1}{N}\tilde{R}_2(r)\right)\Biggr|_{r=1}&\sim -\frac{d}{dr}h_1(r)\Biggr|_{r=1}=-2u. 
\end{align}
Therefore, at the very last stage of the dissipation process, the Page
curve-like behavior is reproduced. Fig.~\ref{renyi_fig} shows the
comparison of $R_2\sps{thermal}$ and $\tilde{R}_2$ for $r\in [0,1]$.
\begin{figure}[tbp]
 \includegraphics[width=7.5cm]{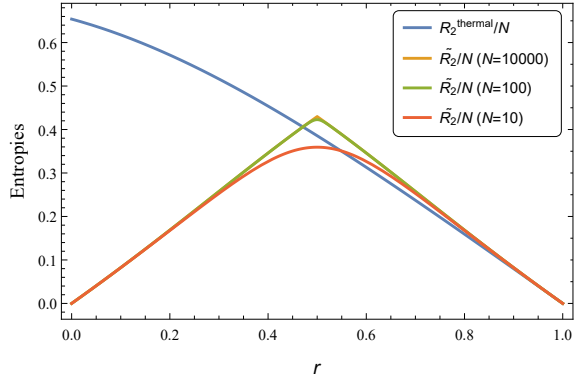}
\caption{Comparison of the typical behavior of the second R\'enyi
 entropy and the thermal R\'enyi entropy for $u=0.4$. At the very last stage of the dissipation process, i.e., $r\approx 1$, $\tilde{R}_2\approx R_2\sps{thermal}$, meaning that the Page curve-like behavior is reproduced.}\label{renyi_fig}
\end{figure}

\subsection{Arbitrary single-qudit process with trivial free Hamiltonian}
In the previous subsection, it has been shown that the second R\'enyi
entropy has a typical behavior for a dissipating qubit system whose
dynamics is described by the multiple-product of the amplitude damping
channel. Then, a natural question arises: is the typicality specific to the
energy-dissipating process? In this subsection, it is proven that for
\textit{any} channel composed of a multiple-product of a single-qudit quantum channel, $R_2$ has a typical behavior if the initial state is randomly chosen from the whole Hilbert space, i.e., $\mathcal{H}_S=\mathcal{H}_A$. This choice of the initial state corresponds to the microcanonical ensemble for a system $A$ with the trivial free Hamiltonian $H=0$ or at infinite temperature. Although this set up may not corresponds to a realistic system, the important implication is that for a wide range of dynamics, there may be a corresponding typical entanglement. 

When $\mathcal{H}_S=\mathcal{H}_A$, $P=\mathbb{I}_{\mathcal{H}_A}$ and $R_\mu=Q_\mu$. Let us introduce an orthonormal basis $\sigma_\mu^{(1\text{-qudit})}$ for the set of linear operators on a single-qudit system. A basis can be constructed by $\sigma_{\mu}\equiv \bigotimes_{n=1}^N\sigma_{\mu_n}^{(1\text{-qudit})}$, where we have identified $\mu=0,\cdots ,d^N-1$ with its $N$-length $d$-ary representation $\mu=\mu_1\mu_2\cdots\mu_N$. Assuming $\mathcal{E}_A=\mathcal{E}^{\otimes N}$ for some single-qudit quantum channel $\mathcal{E}$, $Q_\mu=\bigotimes_{n=1}^NQ_{\mu_n}^{(1\text{-qudit})}$, where we have defined
\begin{align}
 Q_{\mu}^{(1\text{-qudit})}\equiv \sum_i K_i^\dag \sigma_\mu^{(1\text{-qudit})} K_i
\end{align}
by using a set of Kraus operators $\{K_i\}$ of the channel $\mathcal{E}$.
 Therefore, $\alpha_i=a_i^N$ and $\beta_u=b_u^N$ holds for $i=1,2$ and $u=1,\cdots,5$, where
\begin{align}
 a_1\equiv\sum_{\mu=0}^{d^2-1}\mathrm{Tr}\left(Q_\mu^{(1\text{-qudit})}\right)^2, \quad a_2\equiv\sum_{\mu=0}^{d^2-1}\mathrm{Tr}\left(\left(Q_\mu^{(1\text{-qudit})}\right)^2\right),
\end{align}
and
\begin{equation}
 \begin{split}
    b_1&\equiv \sum_{\mu=0}^{d^2-1}\sum_{\nu=0}^{d^2-1}\left(\mathrm{Tr}\left(Q_\mu^{(1\text{-qudit})} Q_\nu^{(1\text{-qudit})}\right)\right)^2,\\
b_2&\equiv\sum_{\mu=0}^{d^2-1}\sum_{\nu=0}^{d^2-1}\mathrm{Tr}\left(Q_\mu^{(1\text{-qudit})} Q_\nu^{(1\text{-qudit})}\right)\mathrm{Tr}\left(Q_\mu^{(1\text{-qudit})}\right)\mathrm{Tr}\left(Q_\nu^{(1\text{-qudit})}\right),\\
 b_3&\equiv \sum_{\mu=0}^{d^2-1}\sum_{\nu=0}^{d^2-1} \mathrm{Tr}\left(Q_\mu^{(1\text{-qudit})} \left(Q_\nu^{(1\text{-qudit})}\right)^2\right)\mathrm{Tr}\left(Q_\mu^{(1\text{-qudit})}\right),\\
 b_4&\equiv \sum_{\mu=0}^{d^2-1}\sum_{\nu=0}^{d^2-1}\mathrm{Tr}\left(\left(Q_\mu^{(1\text{-qudit})}\right)^2 \left(Q_\nu^{(1\text{-qudit})}\right)^2\right),\\
 b_5&\equiv  \sum_{\mu=0}^{d^2-1}\sum_{\nu=0}^{d^2-1}\mathrm{Tr}\left(Q_\mu^{(1\text{-qudit})} Q_\nu^{(1\text{-qudit})} Q_\mu^{(1\text{-qudit})} Q_\nu^{(1\text{-qudit})}\right).
 \end{split}
\end{equation}
Applying Theorem~\ref{thm} for
$\left\{Q_{\mu}^{(1\text{-qudit})}\right\}_{\mu=0}^{d^2-1}$,
$\max_{u=1,\cdots,5}\{|b_u|\}<\max_{i,j=1,2}\{|a_ia_j|\}$ holds, which
implies that the ratio of $\max_{u=1,\cdots,5}\{|\beta_u|\}$ to
$\max_{i,j=1,2}\{|\alpha_i\alpha_j|\}$ is exponentially small in
$N$. Therefore, we have established the following theorem:
\begin{theorem}
 Consider an $N$-qudit system $A$ whose time evolution is described by
 $\mathcal{E}^{\otimes N}$ for some single-qudit channel
 $\mathcal{E}$. Assuming the initial pure state is randomly chosen from
 the whole Hilbert space, the second R\'enyi entropy between the system
 $A$ and its environment is typical when $N\to \infty$.
\end{theorem}

As an example, let us consider the phase damping channel defined by
Kraus operators $K_1\equiv\ket{+}\bra{+}+\sqrt{1-\lambda}\ket{-}\bra{-}$
and $K_2\equiv\sqrt{\lambda}\ket{-}\bra{-}$ for an orthonormal basis
$\{\ket{+},\ket{-}\}$ and $\lambda\in[0,1]$. This channel describes an
information-loss process without loss of energy. As is easily
calculated, $a_1=2$ and $a_2=2(2-\lambda)$. Thus,
$\tilde{R}_2=-\ln{\left((2^N+2^N(2-\lambda)^N)/(2^N(2^N+1))\right)}$. The
plot of $\tilde{R}_2$ is given in Fig. \ref{pha_damp_fig}, showing a qualitatively different behavior from Page-like curves.  
\begin{figure}[tbp]
 \includegraphics[width=7.5cm]{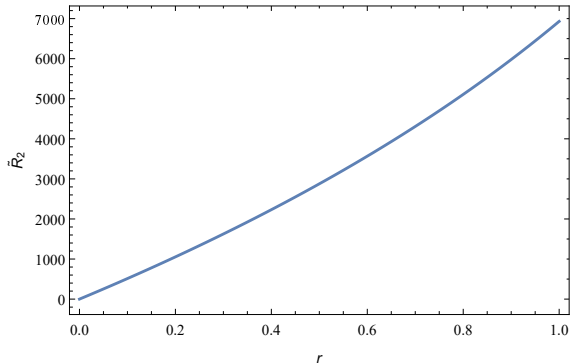}
\caption{The plot of $\tilde{R}_2$ in the phase damping channel for $N=10000$. Since the phase
 damping channel suppresses the off-diagonal component of the initial
 state, the system $A$ becomes more entangled with its environment as
 $\lambda$ becomes large. This behavior is qualitatively different from
 the energy dissipation process described by the amplitude damping
 channel, where the entropy starts to decrease at some point.}\label{pha_damp_fig}
\end{figure}

\section{Conclusions} \label{sec_conclusion}
In this paper, we investigated the typicality of entanglement for dynamical processes.
To investigate the typicality of the $m$-th R\'enyi entropy, we have obtained a formula to calculate the average and the variance for the trace of $m$-th power of a state.
 For a normal thermodynamical system initially prepared according to the
 microcanonical ensemble, the ratio of $\max_{u}\{|\beta_u|\}$ to
 $\max_{i,j}\{|\alpha_i\alpha_j|\}$ gives an important criterion for the
 typicality of the second R\'enyi entropy. We presented two analytically
 tractable cases where the second R\'enyi entropy has a typical
 behavior. The first one was an energy dissipation process in a
 multiple-qubit system, whose typical entanglement is qualitatively similar
 to a Page-like curve. This gives the ``first proof'' of the Page curve
 conjecture as a dynamical process. In the second case, a multiple-qudit
 system was initially prepared as a pure state randomly chosen from the
 whole Hilbert space. It was shown that for dynamics described by a
 multiple-product of a single-qudit quantum channel, there is a typical behavior of the second R\'enyi entropy. This result implies that the typicality is not a specific feature of a energy dissipation process.

The formula for the average and the variance serived in Section
\ref{sec_average_variance} is applicable for any dynamics and
distribution of the initial state, while it is difficult to obtain their
asymptotic behaviors. It is still an open question when entanglement has
a typical behavior in general. It will be also interesting to
investigate the typicality of higher order R\'enyi entropy in future
research. Once one could show typicality for all orders of R\'enyi entropy, the typical behavior of the entanglement entropy can be obtained from the analytical continuation. 

\begin{acknowledgments}
I am grateful to Masahiro Hotta for useful discussions. 
I would like to thank Yasusada Nambu for his help in a numerical calculation which motivated me to start this research.
This work was partially supported by Tohoku University Graduate Program
 on Physics for the Universe (GP-PU), Tohoku University Division for
 Interdisciplinary Advanced Research and Education (DIARE) and Japan Society for the Promotion of Science (JSPS) KAKENHI Grant Number JP18J20057.
\end{acknowledgments}

\appendix

\section{Formula for the formula for ensemble average}\label{ensemble}
\begin{theorem}
For any orthonormal basis $\left\{\ket{i}\right\}_{i=1}^{d_S}$ of $\mathcal{H}_S$ and the coefficients for a Haar-random pure state $\ket{\Psi}\equiv\sum_{i}X_i \ket{i}\in\mathcal{H}$ and a positive integer $m$,
\begin{align}
\overline{\prod_{k=1}^m X_{i^{(k)}}X_{j^{(k)}}{}^*}=\frac{\Gamma(d_S)}{\Gamma(d_S+m)}\sum_{\tau\in S_m}\prod_{k=1}^m\delta_{j^{(k)},i^{(\tau(k))}}\label{average_formula}
\end{align}
, where the overline denotes the integration over the Haar measure for
 the unitary group on $\mathcal{H}_S$, $S_m$ is the symmetric group of degree $m$ and $d_S\equiv\dim{\mathcal{H}_S}$.
\end{theorem}
\begin{proof}
Any pure state in $\mathcal{H}_S$ can be written as $\ket{\Psi}=U\ket{0}=\sum_{i}U_{i,0}\ket{i}$
, where $\ket{0}\in\mathcal{H}_S$ is a unit vector, $U$ is a unitary operator and $U_{i,0}\equiv \braket{i|U|0}$. Taking $U$ as random, we get a random pure state $\ket{\Psi}$.

From the Weingarten calculus formula \cite{Collins}, we have
\begin{align}
 &\overline{\prod_{k=1}^m U_{i^{(k)},j^{(k)}}U_{i^{(k)'},j^{(k)'}}{}^*}=\sum_{\tau_1,\tau_2\in S_m} \prod_{k=1}^m\delta_{i^{(k)'},i^{(\tau_1(k))}}\delta_{j^{(k)'},j^{(\tau_2(k))}}\mathrm{Wg}_{d_S}(\tau_2\tau_1^{-1})\nonumber
\end{align}
, where $\mathrm{Wg}_{d_S}(\tau_2\tau_1^{-1})$ is called the Weingarten function, which depend on $d_S$ and $\tau_2\tau_1^{-1}$. Here we do not need any detail of this function. In the case of random pure state, this formula is simplified as follows:
\begin{align}
 \overline{\prod_{k=1}^m U_{i^{(k)},0}U_{i^{(k)'},0}{}^*}&=\sum_{\tau_1,\tau_2\in S_m} \prod_{k=1}^m\delta_{i^{(k)'},i^{(\tau_1(k))}}\mathrm{Wg}_{d_S}(\tau_2\tau_1^{-1})\nonumber\\
 &=c\sum_{\tau\in S_m} \prod_{k=1}^m\delta_{i^{(k)'},i^{(\tau(k))}}\nonumber
\end{align}
, where we have defined $ c\equiv\sum_{\tau\in S_m}\mathrm{Wg}_{d_S}(\tau)$. 
For notational convenience, let us write $X_{i}\equiv U_{i,0}$. Then, a Haar-random state is written as $\ket{\Psi}=\sum_{i}X_{i}\ket{i}$ and $X_{i}$ satisfies
\begin{align}
 \overline{\prod_{k=1}^m X_{i^{(k)}}X_{i^{(k)'}}{}^*} =c\sum_{\tau\in S_m} \prod_{k=1}^m\delta_{i^{(k)'},i^{(\tau(k))}}\nonumber
\end{align}
The coefficient $c$ is determined by the normalization condition of the
 Haar measure. To calculate $c$, let us consider
 $\mathrm{Tr}\left(\left(\ket{\Psi}\bra{\Psi}\right)^m\right)$. Since
 $\mathrm{Tr}\left(\left(\ket{\Psi}\bra{\Psi}\right)^m\right)=1$ holds for any pure state, we have
\begin{align}
 1&=\overline{1}=\overline{\mathrm{Tr}\left(\left(\ket{\Psi}\bra{\Psi}\right)^m\right)} =\overline{\prod_{k=1}^m\sum_{i^{(k)},i^{(k)'}} X_{i^{(k)}}X_{i^{(k)'}}{}^*\delta_{i^{(k)'},i^{(k+1)}}}\nonumber\\
 &= c\sum_{\tau\in S_m}\prod_{k=1}^m\sum_{i^{(k)},i^{(k)'}} \delta_{i^{(k)'},i^{(\tau(k))}}\delta_{i^{(k)'},i^{(k+1)}}\nonumber\\
 &= c\sum_{\tau\in S_m}\prod_{k=1}^m\sum_{i^{(k)},} \delta_{i^{(\tau(k))},i^{(k+1)}} = c\sum_{\tau\in S_m}\prod_{k=1}^m\sum_{i^{(k)},} \delta_{i^{(\sigma(k))},i^{(k)}} \nonumber\\
 &=c\sum_{k=0}^m\nonumber
\begin{bmatrix}
 m\\
 k
\end{bmatrix}
d_S^k=c\frac{\Gamma(d_S+m)}{\Gamma(d_S)}
\end{align}
, where $\begin{bmatrix}
 m\\
 k
\end{bmatrix}$ is the unsigned Stirling numbers of the first kind, i.e. the number of elements in $S_m$ with $k$ disjoint cycle. Therefore,
\begin{align}
 c=\frac{\Gamma(d_S)}{\Gamma(d_S+m)}\nonumber
\end{align}.
\end{proof}

\section{Proof of Theorem \ref{thm}}\label{thm_proof}
Let us start with the following two lemmas:
 \begin{lemma}\label{lem_pos}
For two positive matrices $M$, $N$, $\mathrm{Tr}\left(MN\right)\leq\mathrm{Tr}\left(M\right)\mathrm{Tr}\left(N\right)$ holds. The
  inequality holds as an equality if and only if one of the following
  conditions is satisfied:
\begin{enumerate}
 \item $M=0$.
 \item $N=0$.
 \item $M=m P$ and $N=n P$ for
  some positive numbers $m,n$ and a one-dimensional projection operator $P$.
\end{enumerate}
 \end{lemma}
\begin{proof} 
By using the eigenvalue decompositions $M=\sum_{i}m_{i}\ket{m_i}\bra{m_i}$ and $N=\sum_i n_i\ket{n_i}\bra{n_i}$, 
\begin{align}
 \mathrm{Tr}\left(MN\right)=\sum_{i,j}m_i n_j\left|\Braket{m_i|n_j}\right|^2,\quad \mathrm{Tr}\left(M\right)\mathrm{Tr}\left(N\right)=\sum_{i,j}m_i n_j.\nonumber
\end{align}
Since $m_i,n_i\geq0$ and $\left|\braket{m_i|n_j}\right|^2\leq 1$,
 $\mathrm{Tr}\left(MN\right)\leq
 \mathrm{Tr}\left(M\right)\mathrm{Tr}\left(N\right)$. The equality holds
 if and only if for any set of $(i,j)$, one of the following conditions
 is satisfied:
 \begin{enumerate}
  \item $m_i=0$.
  \item $n_j=0$.
  \item $\left|\braket{m_i|n_j}\right|=1$.
 \end{enumerate}
If $M$ or $N$ is the zero operator, this condtion is satisfied.
If $M$ and $N$ are not the zero operators, there exists a set of indices $(i_0,j_0)$ such that $m_{i_0}\neq0$ and $n_{i_0}\neq0$. Then $|\braket{m_{i_0}|n_{j_0}}|=1$ holds, and 
  \begin{align}
   0=\braket{n_{j_0}|n_{j_0}}-1=\sum_{i}|\braket{m_i|n_{j_0}}|^2-1=\sum_{i\neq i_0}|\braket{m_i|n_{j_0}}|^2,
  \end{align}
 which implies $m_{i}=0$ for all $i\neq i_0$. Similarly, $n_j=0$ for all
 $j\neq j_0$. Therefore, $M=m_{i_0}P$ and $N=n_{j_0}P$, where $P\equiv
 \ket{m_{i_0}}\bra{m_{i_0}}=\ket{n_{j_0}}\bra{n_{j_0}}$ is a
 one-dimensional projection operator.
\end{proof}
\begin{lemma}\label{coef_ineq}
 For any set of real numbers $\{r_i\}_{i=1}^n$ and Hermite operators $\{R_i\}_{i=1}^n$, it holds
\begin{align}
 \mathrm{Tr}\left(\left(\sum_{i=1}^n r_i R_i\right)^2\right)\leq\left( \sum_{i=1}^n r_i ^2\right)\left(\sum_{j=1}^n\mathrm{Tr}( R_j^2)\right).
\end{align}
The inequality holds as an equality if and only if there exists an Hermite operator $A$ such that for all $i=1,2,\cdots n$, $R_i=r_i A$.
\end{lemma}
\begin{proof}
 Define an orthonormal basis $\{\sigma_\mu\}_{\mu}$ for a set of Hermite operators, satisfying $\mathrm{Tr}\left(\sigma_\mu\sigma_\nu\right)=\delta_{\mu,\nu}$. The Hermite operators can be expanded as $R_i=\sum_{\mu}c_{i,\mu}\sigma_\mu$ by using coefficients $c_{i,\mu}\in\mathbb{R}$. By using the Cauchy--Schwartz inequality, we get
\begin{align}
 \mathrm{Tr}\left(\left(\sum_{i=1}^n r_iR_i\right)^2\right)&=\sum_{\mu}\left(\sum_{i=1}^n r_i c_{i,\mu}\right)^2\nonumber\\
&\leq \sum_{\mu} \left(\sum_{i=1}^nr_i^2\right)\left(\sum_{j=1}^n c_{j,\mu}^2\right)= \left(\sum_{i=1}^nr_i^2\right)\left(\sum_{j=1}^n\mathrm{Tr}( R_j^2)\right).
\end{align}
The equality holds if and only if for all $\mu$, there exists a real number $k_\mu$ such that $c_{i,\mu}=k_\mu r_i$ for $i=1,2,\cdots,n$, i.e.,  $R_i=\sum_{\mu}c_{i,\mu}\sigma_\mu=r_iA$, where $A\equiv \sum_{\mu}k_\mu \sigma_\mu$ is an Hermite operator.
\end{proof}
By using Lemma \ref{lem_pos} and \ref{coef_ineq}, the following theorem is proved:
  \begin{theorem}\label{ineq}
   For any set of Hermite operators $\{R_\mu\}_{\mu}$, the following
   inequality holds:
  \begin{align}
   \max_{i,j}\left\{|\alpha_i\alpha_j|\right\}\geq\max_u\left\{|\beta_u|\right\},\label{thm1_eq}
  \end{align}
   where
   \begin{align}
   \alpha_1&\equiv \sum_{\mu}\mathrm{Tr}\left(R_\mu\right)^2,\,\alpha_2\equiv\sum_{\mu}\mathrm{Tr}\left(R_\mu^2\right) ,\nonumber\\
    \beta_1&\equiv \sum_{\mu,\nu}\left(\mathrm{Tr}\left(R_\mu
    R_\nu\right)\right) ^2,\,    \beta_2\equiv \sum_{\mu,\nu} \mathrm{Tr}\left(R_\mu
    R_\nu\right)\mathrm{Tr}\left(R_\mu\right)\mathrm{Tr}\left(R_\nu\right),\nonumber\\
    \beta_3&\equiv \sum_{\mu,\nu}  \mathrm{Tr}\left(R_\mu
    R_\nu^2\right)\mathrm{Tr}\left(R_\mu\right) ,\,
    \beta_4\equiv  \sum_{\mu,\nu} \mathrm{Tr}\left(R_\mu^2 R_\nu^2\right) ,\,
    \beta_5\equiv  \sum_{\mu,\nu} \mathrm{Tr}\left(R_\mu R_\nu R_\mu R_\nu\right). \nonumber
   \end{align}
   The inequality holds as an
   equality if and only if there exists an Hermite operator $A$
   satisfying $\mathrm{Tr}\left(A^2\right)\geq
   \mathrm{Tr}\left(A\right)^2$ such that $\forall \mu$, $\exists r_\mu\in\mathbb{R}$, $R_\mu=r_\mu A$.
  \end{theorem}
  \begin{proof}
   Without loss of generality, one can assume $R_\mu\neq 0$ for all $\mu$
   since a zero operator does not contribute all $\alpha$s and $\beta$s. Let us consider
   an upperbound for each $\beta_u$ one by one.
   \begin{itemize}
\item $  \beta_1=\sum_{\mu,\nu}\left(\mathrm{Tr}\left(R_\mu R_\nu\right) \right)^2$:\\
For any set of $(\mu,\nu)$,
   \begin{align}
    \left(\mathrm{Tr}\left(R_\mu R_\nu\right) \right)^2\leq \mathrm{Tr}\left(R_\mu^2\right)\mathrm{Tr}\left(R_\nu^2\right),
   \end{align}
   where we have used the Cauchy-Schwartz inequality for the
      Hilbert--Schmidt inner product. The inequality holds as an
      equality if and only if $\exists \alpha\indices{_\mu^\nu}\in
      \mathbb{R}$ s.t. $R_\mu=\alpha\indices{_\mu^\nu}R_\nu$. Since $R_\mu\neq0$
      for all $\mu$, $\exists r_\mu\in\mathbb{R}\backslash \{0\}$ s.t. $R_\mu=r_\mu A$ for some Hermite operator $A$. Then, $\left(\mathrm{Tr}\left(R_\mu\right)\right)^2=r_\mu^2\left(\mathrm{Tr}\left(A\right)\right)^2$ and $\mathrm{Tr}\left(R_\mu^2\right)=r_\mu^2\mathrm{Tr}\left(A^2\right)$.
$\sum_{\mu,\nu}\mathrm{Tr}\left(R_\mu^2\right)\mathrm{Tr}\left(R_\nu^2\right)=\max_{i,j}\left\{|\alpha_i\alpha_j|\right\}$
      if and only if $\mathrm{Tr}\left(R_\mu^2\right)\geq
      \left(\mathrm{Tr}\left(R_\mu\right)\right)^2$, or equivalently
      $r_\mu^2(\mathrm{Tr}\left(A^2\right)-\mathrm{Tr}\left(A\right)^2)\geq
      0$ holds for all
      $\mu$. Thus, $\mathrm{Tr}\left(A^2\right)\geq\mathrm{Tr}\left(A\right)^2$.
 \item $ \beta_2=\sum_{\mu,\nu}\mathrm{Tr}\left(R_\mu
       R_\nu\right)\mathrm{Tr}\left(R_\mu\right)\mathrm{Tr}\left(R_\nu\right)$:\\
       By using Lemma \ref{coef_ineq}, 
   \begin{align}
    \sum_{\mu,\nu}\mathrm{Tr}\left(R_\mu
    R_\nu\right)\mathrm{Tr}\left(R_\mu\right)\mathrm{Tr}\left(R_\nu\right)=\mathrm{Tr}\left(\left(\sum_{\mu}\mathrm{Tr}\left(R_\mu\right)
    R_\mu\right)^2\right)\leq\sum_{\mu,\nu}\mathrm{Tr}\left(R_\mu\right)^2\mathrm{Tr}\left(R_\nu^2\right)\label{munu_mu_nu}.
   \end{align}
   This inequality holds as an equality if and only if there exist an
       Hermite operator $A$ s.t. $R_\mu=\mathrm{Tr}\left(R_\mu\right)A$
       for all $\mu$. Since we have assumed $R_\mu\neq0$, this conditioin implies
       $\mathrm{Tr}\left(A\right)=1$. Noting
       $\sum_{\mu,\nu}\mathrm{Tr}\left(R_\mu\right)^2\mathrm{Tr}\left(R_\nu^2\right)(=\alpha_1\alpha_2)=\max_{i,j}\left\{|\alpha_i\alpha_j|\right\}$
       if and only if
       $\sum_{\mu}\mathrm{Tr}\left(R_\mu^2\right)=\sum_{\mu}\mathrm{Tr}\left(R_\mu\right)^2$,
       it follows that $\mathrm{Tr}\left(A^2\right)=1(=\mathrm{Tr}\left(A\right)^2)$.
    \item $ \beta_3=\sum_{\mu,\nu}\mathrm{Tr}\left(R_\mu R_\nu^2\right)\mathrm{Tr}\left(R_\mu\right)$:\\
By using the Cauchy--Schwartz inequality,
\begin{align}
\left|\sum_{\mu,\nu}\mathrm{Tr}\left(R_\mu
 R_\nu^2\right)\mathrm{Tr}\left(R_\mu\right)\right|&=\left|\sum_{\mu}\mathrm{Tr}\left(
 R_\mu \sum_\nu R_\nu^2\right)\mathrm{Tr}\left(R_\mu\right)\right|\nonumber\\
 &\leq  \sqrt{\sum_{\mu}\left(\mathrm{Tr}\left(R_\mu\right)\right)^2\sum_{\rho}\left(\mathrm{Tr}\left(R_\rho \sum_{\nu}R_\nu^2\right)\right)^2}.
\end{align}
By using the Cauchy--Schwartz inequality, 
\begin{align}
 \left(\mathrm{Tr}\left(R_\rho \sum_{\nu}R_\nu^2\right)\right)^2&\leq\mathrm{Tr}\left(\left(\sum_\nu R_\nu^2\right)^2\right)\mathrm{Tr}\left(R_\rho^2\right)=\sum_{\mu,\nu}\mathrm{Tr}\left(R_\mu^2 R_\nu^2\right)\mathrm{Tr}\left(R_\rho^2\right)\nonumber\\
 &\leq \sum_{\mu,\nu}\mathrm{Tr}\left(R_\mu^2\right)\mathrm{Tr}\left(R_\nu^2\right)\mathrm{Tr}\left(R_\rho^2\right)
\end{align}
holds for all $\rho$. On the second line, we have used Lemma
	  \ref{lem_pos}. The second inequality holds as an equality if and only if
       $R_\mu=r_\mu P$ for some $r_\mu\in\mathbb{R}$ and a one-dimensional projector $P$.
Thus, 
\begin{align}
 \left|\sum_{\mu,\nu}\mathrm{Tr}\left(R_\mu R_\nu^2\right)\mathrm{Tr}\left(R_\mu\right)\right|&\leq \sqrt{\left(\sum_{\mu}\left(\mathrm{Tr}\left(R_\mu\right)\right)^2\right)\left(\sum_{\nu}\mathrm{Tr}\left(R_\nu^2\right)\right)^3}\nonumber\\
&\leq \max{\left\{\alpha_1^2,\alpha_2^2\right\}}.
\end{align}
 \item $ \beta_4=\sum_{\mu,\nu}\mathrm{Tr}\left(R_\mu^2R_\nu^2\right)$\\
       Since $R_\nu^2$ is a positive operator,
   \begin{align}
    \mathrm{Tr}\left(R_\mu^2R_\nu^2\right)\leq \mathrm{Tr}\left(R_\mu^2\right)\mathrm{Tr}\left(R_\nu^2\right).
   \end{align}
    Since we have assuumed $R_\mu\neq 0$, the inequality holds as an equality for any pair of $(\mu,\nu)$ if
       and only if there exists one-dimensional projection operator $P$
       s.t. $\forall\mu$, $\exists a_\mu\in\mathbb{R}_{>0}$
       s.t. $R_\mu^2=a_\mu P$. This condition is equivalent to
       $R_\mu=r_\mu P$, where $r_\mu=\pm \sqrt{a_\mu}$. 

 \item $ \beta_5=\sum_{\mu,\nu}\mathrm{Tr}\left(R_\mu R_\nu R_\mu R_\nu\right)$\\
   By using a commutator defined by $\left[A,B\right]\equiv AB-BA$, 
   \begin{align}
 \mathrm{Tr}\left(R_\mu R_\nu R_\mu R_\nu\right)&=\mathrm{Tr}\left(R_\mu^2 R_\nu^2\right)+\mathrm{Tr}\left(R_\mu R_\nu\left[R_\mu,R_\nu\right]\right)
   \end{align}
       holds.
Since $i\left[Q_\mu,Q_\nu\right]$ is Hermite, 
\begin{align}
 \sum_{\mu,\nu}\mathrm{Tr}\left(R_\mu R_\nu\left[R_\mu,R_\nu\right]\right)=-\frac{1}{2}\sum_{\mu,\nu}\mathrm{Tr}\left(\left(i\left[R_\mu,R_\nu\right]\right)^2\right)\leq 0.
\end{align}
Thus, 
\begin{align}
 \sum_{\mu,\nu}\mathrm{Tr}\left(R_\mu R_\nu R_\mu R_\nu\right)\leq \sum_{\mu,\nu}\mathrm{Tr}\left(R_\mu^2R_\nu^2\right).
\end{align}
       By using anticommutator $\left\{A,B\right\}\equiv
       AB+BA$, a similar calculation shows
        \begin{align}
	 \mathrm{Tr}\left(R_\mu R_\nu R_\mu R_\nu\right)=-\mathrm{Tr}\left(R_\mu^2 R_\nu^2\right)+\mathrm{Tr}\left(R_\mu \left\{R_\mu,R_\nu\right\}R_\nu\right)
	\end{align}
       and
       \begin{align}
	\sum_{\mu,\nu}\mathrm{Tr}\left(R_\mu\left\{R_\mu,R_\nu\right\}R_\nu\right)=\frac{1}{2}\sum_{\mu,\nu}\mathrm{Tr}\left(\left\{R_\mu,R_\nu\right\}^2\right)\geq 0,
       \end{align}
       which implies
       \begin{align}
	\sum_{\mu,\nu}\mathrm{Tr}\left(R_\mu R_\nu R_\mu R_\nu\right)\geq-\sum_{\mu,\nu}\mathrm{Tr}\left(R_\mu^2 R_\nu^2\right).
       \end{align}
       Therefore, $|\beta_5|\leq |\beta_4|$.
   \end{itemize}
  \end{proof}
Let us define a set of operators $\{R_\mu\}_\mu$ as $R_\mu\equiv
P\left(\sum_{i}K_i^\dag\sigma_\mu K_i\right)P$ by using a set of Kraus
operators $\{K_i\}_i$ and a projection operator $P$ onto the sub-Hilbert
space $\mathcal{H}_S$. Since Kraus
operators satisfy $\sum_i K_i^\dag K_i =\mathbb{I}_{\mathcal{H}_A}$, we
get $R_0=\frac{1}{\sqrt{d_A}}P\mathbb{I}_A
P=\frac{1}{\sqrt{d_A}}P$. Assume the equality in Eq.~(\ref{thm1_eq})
holds, then $A=kP$ for some $k\in\mathbb{R}\backslash \{0\}$. Since
$\mathrm{Tr}\left(A^2\right)=k^2 d_S$ and
$\mathrm{Tr}\left(A\right)^2=k^2 d_S^2$, the conditioin
$\mathrm{Tr}\left(A^2\right)\geq \mathrm{Tr}\left(A\right)^2$ implies
$d_S=1$, which concludes the proof of Theorem \ref{thm}.

\section{The asymptotic behavior of the Hypergeometric function}
\label{asy_appendix}
In this section, we derive the asymptotic behavior of the Hypergeometric
function ${}_2F_1\left(-uN,-uN;(1-2u)N+1;z\right)$ for $u\in[0,1/2)$. It is known that the
Hypergeometric function ${}_2F_1\left(a,b;c;z\right)$ is the solution
for the following differential equation satisfying $f(z=0)=1$:
\begin{align}
 z(1-z)\frac{d^2}{dz^2}f+(c-(a+b+1)z)\frac{d}{dz}f-abf=0.
\end{align}

Defining a function $g(z)\equiv \frac{1}{N}\ln{f(z)}$, we get
\begin{align}
 f^{-1}\frac{d}{dz}f(z)=N\frac{d}{dz}g(z),\quad f^{-1}\frac{d^2}{dz^2}f(z)=\left(N\frac{d}{dz}g(z)\right)^2+N\frac{d^2}{dz^2}g(z).
\end{align}
Therefore, for $N\gg1$, $g(z)$ satisfies
\begin{align}
 z(1-z)\left(\frac{d}{dz}g(z)\right)^2+(\gamma-(\alpha+\beta)z)\frac{d}{dz}g-\alpha\beta=0,
\end{align}
where we have assumed $a,b,c\in\mathbb{R}$ and they have the form of
$a=\alpha N +a'$, $b=\beta N +b'$ and $c=\gamma N+c'$ with
$N$-independent constants $a',\,b',\,c',\,\alpha,\,\beta$, and $\gamma$. Then,
$\frac{d}{dz}g(z)$ satisfies
\begin{align}
 \frac{d}{dz}g(z)&=\frac{-(\gamma-(\alpha+\beta)z)\pm\sqrt{\left(\gamma-(\alpha+\beta)z\right)^2+4\alpha\beta
 z(1-z)}}{2z(1-z)}\nonumber\\
 &=\frac{-(\gamma-(\alpha+\beta)z)\pm|\gamma|\sqrt{1+s z+ tz^2}}{2z(1-z)} ,
\end{align}
where we have defined $s\equiv
\frac{4\alpha\beta-2\gamma(\alpha+\beta)}{\gamma^2}$ and $t\equiv
\frac{(\alpha-\beta)^2}{\gamma^2}$. One can confirm that
\begin{align}
 \frac{\sqrt{1+sz
 +tz^2}}{z(1-z)}&=\frac{\sqrt{1+s+t}}{1-z}+\frac{1}{z}-\frac{s+\frac{s+2tz}{\sqrt{1+sz+tz^2}}}{2+sz+2\sqrt{1+s
 z+tz^2}}-\sqrt{t}\frac{\left(2t+\frac{\sqrt{t}\left(s+2tz\right)}{\sqrt{1+
 sz +tz^2}}\right)}{s+2t z+2\sqrt{t}\sqrt{1+sz+tz^2}}\nonumber\\
 &\qquad\qquad +\frac{\sqrt{1+s+t}\left(s+2t+\frac{\sqrt{1+s+t}\left(s+2tz\right)}{\sqrt{1+sz+tz^2}}\right)}{2+s+sz+2tz+2\sqrt{1+s+t}\sqrt{1+sz+tz^2}}, 
\end{align}
meaning that
\begin{align}
 &\int dz  \frac{\sqrt{1+sz
 +tz^2}}{z(1-z)}\nonumber\\
 &=-\sqrt{1+s+t}\ln{\left(1-z\right)}+\ln{z}-\ln{\left(2+s
 z+2\sqrt{1+s z+tz^2}\right)}-\sqrt{t}\ln{\left(s+2tz +2\sqrt{t}\sqrt{1+sz+tz^2}\right)}\nonumber\\
 &\qquad\qquad+\sqrt{1+s+t}\ln{\left(2+s+sz+2tz+2\sqrt{1+s+t}\sqrt{1+sz+tz^2}\right)}.
\end{align}
Imposing the boundary condition $g(z=0)=0$, we get
\begin{align}
 &g(z)\nonumber\\
 &=
 \frac{1}{2}\left(\gamma-(\alpha+\beta)-\gamma\sqrt{1+s+t}\right)\ln{\left(1-z\right)}\nonumber\\
 &\qquad\qquad+\frac{\gamma}{2}\left(-\ln{\left(\frac{2+sz+2\sqrt{1+sz+tz^2}}{4}\right)}-\sqrt{t}\ln{\left(\frac{s+2tz+2\sqrt{t}\sqrt{1+sz+tz^2}}{s+2\sqrt{t}}\right)}\right.\nonumber\\
 &\qquad\qquad\left.+\sqrt{1+s+t}\ln{\left(\frac{2+s+sz+2tz+2\sqrt{1+s+t}\sqrt{1+sz+tz^2}}{2+s+2\sqrt{1+s+t}}\right)}\right) \label{eq_asy}.
\end{align}
Since $\sqrt{1+s+t}=\left|\frac{\gamma-(\alpha+\beta)}{\gamma}\right|$,
the divergent term at $z=1$ vanishes if
$\frac{\gamma-(\alpha+\beta)}{\gamma}\geq 0$. For $\alpha=\beta=-u$ and
$\gamma=1-2u$, $\frac{\gamma-(\alpha+\beta)}{\gamma}=\frac{1}{1-2u}>0$
holds for $u\in[0,1/2)$. Substituting $s=\frac{4u(1-u)}{(1-2u)^2}$ and
$t=0$ into Eq.~(\ref{eq_asy}),
${}_2F_1\left(-uN,-uN;(1-2u)N+1;z\right)\sim \exp{\left(Ng_u(z)\right)}$
in the limit of $N\to\infty$, where
\begin{align}
 g_u(z)&\equiv
 -\left(\frac{1}{2}-u\right)\ln{\left(\frac{2+\frac{4u(1-u)}{(1-2u)^2}z+2\sqrt{1+\frac{4u(1-u)}{(1-2u)^2}z}}{4}\right)}\nonumber\\
 &\qquad\qquad +\frac{1}{2}\ln{\left(\frac{2+\frac{4u(1-u)}{(1-2u)^2}(1+z)+\frac{2}{1-2u}\sqrt{1+\frac{4u(1-u)}{(1-2u)^2}z}}{4\left(\frac{1-u}{1-2u}\right)^2}\right)}.
\end{align}
Since
\begin{align}
&
 \frac{2+\frac{4u(1-u)}{(1-2u)^2}(1+z)+\frac{2}{1-2u}\sqrt{1+\frac{4u(1-u)}{(1-2u)^2}z}}{4\left(\frac{1-u}{1-2u}\right)^2}> \frac{2+\frac{4u(1-u)}{(1-2u)^2}z+2\sqrt{1+\frac{4u(1-u)}{(1-2u)^2}z}}{4} 
\end{align}
and $1/2-u< 1/2$ holds for $u\in(0,1/2)$ and $z> 0$, $g_u(z)>0$ holds, meaning that $\exp{\left(Ng_u(z)\right)}$ is exponentially large in $N$
for any $u\in(0,1/2)$ and $z> 0$.

\end{document}